\documentclass{scrartcl}
\usepackage{pslatex}
\usepackage[latin1]{inputenc}
\usepackage[T1]{fontenc}
\usepackage{graphicx}
\usepackage[ruled, linesnumbered]{algorithm2e}
\usepackage{enumitem}
\setlist[enumerate]{label=(\alph*)}
\DeclareMathAlphabet{\mathpzc}{OT1}{pzc}{m}{it}
\usepackage{amsmath}
\usepackage{amssymb}
\usepackage[mathscr]{eucal}
\usepackage[thmmarks,amsmath]{ntheorem}
\usepackage{braket}
\DeclareMathOperator{\proj}{proj}
\DeclareMathOperator{\finished}{finished}
\DeclareMathOperator{\solver}{solver}
\DeclareMathOperator{\bisection}{Bisection}
\DeclareMathOperator{\newton}{Newton}
\DeclareMathOperator{\halley}{Halley}
\DeclareMathOperator{\newtonsqr}{NewtonSqr}
\DeclareMathOperator{\lo}{lo}
\DeclareMathOperator{\up}{up}
\DeclareMathOperator{\aux}{auxiliary}
\newcommand{\B}{\mathbb{B}}
\newcommand{\N}{\mathbb{N}}
\newcommand{\R}{\mathbb{R}}
\newcommand{\X}{\mathcal{X}}
\newcommand{\lag}{\mathscr{L}}
\newcommand{\shrink}{\mathcal{S}}
\newcommand{\scndmax}{\text{2nd-}\max}
\newcommand{\sy}{\mathrm{sum}_{\tilde{y}}}
\newcommand{\scpy}{\mathrm{scp}_{p,\,y}}
\newcommand{\intervalcc}[2]{\left[#1,\ #2\right]}
\newcommand{\intervaloo}[2]{\left(#1,\ #2\right)}
\newcommand{\intervalco}[2]{\left[#1,\ #2\right)}
\newcommand{\discint}[2]{\{#1,\dotsc,#2\}}
\newcommand{\inint}[2]{\in\discint{#1}{#2}}
\newcommand{\abs}[1]{\left\vert #1 \right\vert}
\newcommand{\norm}[1]{\left\Vert#1\right\Vert}
\newcommand{\bnorm}[1]{\big\Vert#1\big\Vert}
\newcommand{\transp}{^T}
\newcommand{\musteq}{\overset{!}{=}}
\renewcommand{\emptyset}{\varnothing}

\deffootnote[1.25em]{1.25em}{1.25em}{\textsuperscript{\thefootnotemark}}

\begin{document}
\begin{center}
{\bf\sffamily\huge Efficient Sparseness-Enforcing Projections}\bigskip\bigskip

{\sffamily\Large Markus Thom\footnote{driveU / Institute of Measurement, Control and Microtechnology, Ulm University, Ulm, Germany} and G\"{u}nther Palm\footnote{Institute of Neural Information Processing, Ulm University, Ulm, Germany\\\noindent E-mail addresses: {\tt markus.thom@uni-ulm.de}, {\tt guenther.palm@uni-ulm.de}}}
\end{center}

\begin{abstract}\noindent
{\bf Abstract.}
We propose a linear time and constant space algorithm for computing Euclidean projections onto sets on which a normalized sparseness measure attains a constant value.
These non-convex target sets can be characterized as intersections of a simplex and a hypersphere.
Some previous methods required the vector to be projected to be sorted, resulting in at least quasilinear time complexity and linear space complexity.
We improve on this by adaptation of a linear time algorithm for projecting onto simplexes.
In conclusion, we propose an efficient algorithm for computing the product of the gradient of the projection with an arbitrary vector.
\end{abstract}

\section{Introduction}
In a great variety of classical machine learning problems, sparse solutions are appealing because they provide more efficient representations compared to non-sparse solutions.
Several formal sparseness measures have been proposed in the past and their properties have been thoroughly analyzed~\cite{Hurley2009}.
One remarkable sparseness measure is the normalized ratio of the $L_1$ norm and the $L_2$ norm of a vector, as originally proposed by~\cite{Hoyer2004}:
\begin{displaymath}
  \sigma\colon\R^n\setminus\set{0}\to\intervalcc{0}{1}\text{,}\qquad x\mapsto\frac{\sqrt{n} - \frac{\norm{x}_1}{\norm{x}_2}}{\sqrt{n}-1}\text{.}
\end{displaymath}
Here, higher values of $\sigma$ indicate more sparse vectors.
The extreme values of $0$ and $1$ are achieved for vectors where all entries are equal and vectors where all but one entry vanish, respectively.
Further, $\sigma$ is scale-invariant, that is $\sigma(\alpha x) = \sigma(x)$ for all $\alpha\neq 0$ and all $x\in\R^n\setminus\set{0}$.

The incorporation of explicit sparseness constraints to existing optimization problems while still being able to efficiently compute solutions to them was made possible by~\cite{Hoyer2004} through proposition of an operator, which computes the Euclidean projection onto sets on which $\sigma$ attains a desired value.
In other words, given a target degree of sparseness $\sigma^*\in\intervaloo{0}{1}$ with respect to $\sigma$, numbers $\lambda_1,\lambda_2 > 0$ can be derived such that $\sigma\equiv\sigma^*$ on the non-convex set
\begin{displaymath}
  D := \Set{s\in\R_{\geq 0}^n | \norm{s}_1 = \lambda_1\text{ and }\norm{s}_2 = \lambda_2}\text{.}
\end{displaymath}
Clearly, either of $\lambda_1$ and $\lambda_2$ has to be fixed to a pre-defined value, for example by setting $\lambda_2 := 1$ for achieving normalized vectors, as only their ratio is important in the definition of $\sigma$.
By restricting possible solutions to certain optimization problems to lie in $D$, projected gradient descent methods~\cite{Bertsekas1999} can be used to achieve solutions that fulfill explicit sparseness constraints.

The projection operator of~\cite{Hoyer2004} was motivated by geometric ideas, in such that the intersection of hyperplanes, hyperspheres and the non-negative orthant were considered and a procedure of alternating projections was proposed.
This procedure is known to produce correct projections when applied to the intersection of convex sets~\cite{Deutsch2001}.
In the non-convex setup considered here, it is not clear in the first place whether the method also computes correct projections.
The results of~\cite{Theis2005}, however, show that alternating projections also work for the sparseness projection, and that the projection onto $D$ is unique almost everywhere.

It was further noted recently that the method of Lagrange multipliers can also be used to derive an implicit, compact representation of the sparseness projection~\cite{Potluru2013}.
The algorithm proposed there needs to sort the vector that is to be projected and remember the sorting permutation, resulting in a computational complexity that is quasilinear and a space complexity that is linear in the problem dimensionality $n$.
In this work, we provide a detailed derivation of the results of~\cite{Potluru2013}.
Then, by transferring the ideas of~\cite{Liu2009a} to efficiently compute projections onto simplexes we use the implicit representation to propose a linear time and constant space algorithm for computing projections onto $D$.
Ultimately, we propose an algorithm that efficiently computes the product of the gradient of the projection onto $D$ with an arbitrary vector.

\section{Notation and Prerequisites}
We denote the set of Boolean values with $\B$, the real numbers with $\R$ and the $n$-dimensional Euclidean space with $\R^n$.
Subscripts for elements from $\R^n$ denote individual coordinates.
All entries of the vector $e\in\R^n$ are unity.
$\R^{n\times n}$ is the ring of matrices with $n$ rows and $n$ columns, and $E_n\in\R^{n\times n}$ is the identity matrix.
It is well-known that the $L_1$ norm and the $L_2$ norm are equivalent in the topological sense~\cite{Laub2004}:
\begin{remark}
\label{rem:norms_equivalent}
For all $x\in\R^n$, we have that $\norm{x}_2 \leq \norm{x}_1 \leq \sqrt{n}\norm{x}_2$.
If $x$ is sparsely populated, then the latter inequality can be sharpened to $\norm{x}_1 \leq \sqrt{d}\norm{x}_2$, where $d := \norm{x}_0 \leq n$ denotes the number of non-vanishing entries in $x$.
\end{remark}
Therefore $\lambda_2 < \lambda_1 < \sqrt{n}\lambda_2$ must hold for the target norms to achieve a sparseness of $\sigma^*\in\intervaloo{0}{1}$.
The projection onto a set contains all points with infimal distance to the projected vector~\cite{Deutsch2001}:
\begin{definition}
\label{dfn:projection}
Let $x\in\R^n$ and $\emptyset\neq M\subseteq\R^n$.
Then every point in
\begin{displaymath}
  \proj_M(x) := \set{y\in M | \norm{y - x}_2 \leq \norm{z - x}_2\text{ for all }z\in M}
\end{displaymath}
is called \emph{Euclidean projection} of $x$ onto $M$.
If there is exactly one point $y$ in $\proj_M(x)$, then $y = \proj_M(x)$ is written for abbreviation.
\end{definition}
We further note that projections onto permutation-invariant sets are order-preserving:
\begin{proposition}
\label{prop:orderpres}
Let $\emptyset\neq M\subseteq\R^n$ such that $P_\tau x\in M$ for all $x\in M$ and all permutation matrices $P_\tau\in\R^{n\times n}$.
Let $x\in\R^n$ and $p\in\proj_M(x)$.
Then $x_i > x_j$ implies $p_i \geq p_j$ for all $i,j\inint{1}{n}$.
\end{proposition}
Throughout the paper, we assume that the input vector to the projection operator is chosen such that the projection onto $D$ is unique.
As has been shown by~\cite{Theis2005}, this is fulfilled by almost all $x\in\R^n\setminus\set{0}$ and is thus no restriction in practice.

\section{Implicit Representation of the Projection}
As noted by~\cite{Potluru2013}, the method of Lagrange multipliers can be used to derive an implicit representation of the projection onto $D$.
To make this paper as self-contained as possible, we include an elaborate derivation of their result.

\begin{lemma}
\label{lem:representation}
Let $x\in\R_{\geq 0}^n\setminus D$ such that $\proj_D(x)$ is unique.
Then there exist unique numbers $\alpha\in\R$ and $\beta\in\R_{>0}$ such that $\proj_D(x) = \max\big(\frac{1}{\beta}\left(x - \alpha\cdot e\right),\ 0\big)$.
\end{lemma}
\begin{proof}
We want to find a point $p\in D$ such that the Euclidean distance $\norm{p - x}_2$ is minimized.
Such a point is guaranteed to exist by the Weierstra\ss\ extreme value theorem.
The constrained optimization problem leads to the Lagrangian $\lag\colon\R^n\times\R\times\R\times\R_{\geq 0}^n\to\R$,
\begin{displaymath}
  (p,\ \alpha,\ \beta,\ \gamma)\mapsto \tfrac{1}{2}\norm{p - x}_2^2 + \alpha\left(\norm{p}_1 - \lambda_1\right) + \tfrac{\beta - 1}{2}\big(\norm{p}_2^2 - \lambda_2^2\big) - \gamma\transp p\text{,}
\end{displaymath}
where the multiplier $\beta$ was linearly transformed for notational convenience.
By taking the derivative for $p$ and setting it to zero we obtain
\begin{displaymath}
  \frac{\partial\lag}{\partial p_i} = \beta p_i - x_i + \alpha - \gamma_i \musteq 0\text{, and hence }
  p_i = \frac{x_i - \alpha + \gamma_i}{\beta}\text{ for all }i\inint{1}{n}\text{.}
\end{displaymath}
The complementary slackness condition, $\gamma_ip_i = 0$ for all $i\inint{1}{n}$, must be satisfied in a local minimum of $\lag$.
Hence $p_i > 0$ implies $\gamma_i = 0$ and $p_i = 0$ implies $\gamma_i\geq 0$ for all $i\inint{1}{n}$.
Let $I := \set{i\inint{1}{n} | p_i > 0}$ denote the set of coordinates in which $p$ does not vanish, and let $d := \abs{I}$ denote its cardinality.
We have $d\geq 2$, because $d = 0$ is impossible due to $\lambda_1,\lambda_2 > 0$ and $d = 1$ is impossible because $\lambda_1\neq \lambda_2$.
Further, $\gamma_i = 0$ for all $i\in I$ from the complementary slackness condition.
Let $\tilde{x}\in\R_{\geq 0}^d$ be the vector with all entries from $x$ with index in $I$, that is when $I = \discint{i_1}{i_d}$ then $\tilde{x}\transp = \left(x_{i_1},\ \dots,\ x_{i_d}\right)$.
Note that because all entries of $p$ and $x$ are non-negative, the sum over their entries is identical to their $L_1$ norm.
By taking the derivative of the Lagrangian for $\alpha$ and setting it to zero we have that
\begin{displaymath}
  \lambda_1
  = \norm{p}_1
  = \sum_{i\in I}p_i
  = \sum_{i\in I}\tfrac{1}{\beta}\left(x_i - \alpha\right)
  = \tfrac{1}{\beta}\left(\norm{\tilde{x}}_1 - d\alpha\right)\text{.}
\end{displaymath}
Analogously, taking the derivative for $\beta$ and setting it to zero yields
\begin{displaymath}
  \lambda_2^2
  = \norm{p}_2^2
  = \tfrac{1}{\beta^2}\sum_{i\in I}\left(x_i - \alpha\right)^2
  = \tfrac{1}{\beta^2}\big(\norm{\tilde{x}}_2^2 - 2\alpha\norm{\tilde{x}}_1 + d\alpha^2\big)\text{.}
\end{displaymath}
By squaring the expression for $\lambda_1$ and dividing by $\lambda_2^2$ we get
\begin{displaymath}
  \frac{\lambda_1^2}{\lambda_2^2}
  = \frac{\norm{\tilde{x}}_1^2 - 2d\alpha\norm{\tilde{x}_1} + d^2\alpha^2}{\norm{\tilde{x}}_2^2 - 2\alpha\norm{\tilde{x}}_1 + d\alpha^2}\text{,}
\end{displaymath}
which leads to the quadratic equation
\begin{displaymath}
  0
  =   \underbrace{d\left(d - \tfrac{\lambda_1^2}{\lambda_2^2}\right)}_{=: a}\cdot\alpha^2
    + \underbrace{2\norm{\tilde{x}}_1\left(\tfrac{\lambda_1^2}{\lambda_2^2} - d\right)}_{=: b}\cdot\alpha
    + \underbrace{\left(\norm{\tilde{x}}_1^2 - \tfrac{\lambda_1^2}{\lambda_2^2}\norm{\tilde{x}}_2^2\right)}_{=: c}\text{.}
\end{displaymath}
Before considering the discriminant of this equation, we first note that $d\norm{\tilde{x}}_2^2 - \norm{\tilde{x}}_1^2 \geq 0$ with Remark~\ref{rem:norms_equivalent}.
As $p$ exists by the Weierstra\ss\ extreme value theorem and has by definition $d$ non-zero entries, we also have that $d - \frac{\lambda_1^2}{\lambda_2^2} \geq 0$ using Remark~\ref{rem:norms_equivalent}.
Thus we obtain
\begin{align*}
  D
  &:= b^2 - 4ac
  =   4\norm{\tilde{x}}_1^2\left(d - \tfrac{\lambda_1^2}{\lambda_2^2}\right)^2
    - 4d\left(d - \tfrac{\lambda_1^2}{\lambda_2^2}\right)\left(\norm{\tilde{x}}_1^2 - \tfrac{\lambda_1^2}{\lambda_2^2}\norm{\tilde{x}}_2^2\right)\\
  &= 4\tfrac{\lambda_1^2}{\lambda_2^2}\left(d - \tfrac{\lambda_1^2}{\lambda_2^2}\right)\left(d\norm{\tilde{x}}_2^2 - \norm{\tilde{x}}_1^2\right)
  \geq 0\text{,}
\end{align*}
so $\alpha$ must be a real number.
Solving the equation leads to two possible values for $\alpha$:
\begin{displaymath}
  \alpha\in\Set{\frac{-b\pm\sqrt{D}}{2a}}
  = \Set{\frac{1}{d}\left(\norm{\tilde{x}}_1\pm \lambda_1\sqrt{\frac{d\norm{\tilde{x}}_2^2 - \norm{\tilde{x}}_1^2}{d\lambda_2^2 - \lambda_1^2}}\right)}\text{.}
\end{displaymath}
We first assume that $\alpha$ is the number that arises from the "$+$" before the square root.
From $\lambda_1 = \norm{p}_1$ we then obtain
\begin{displaymath}
  \beta
  = \tfrac{1}{\lambda_1}\left(\norm{\tilde{x}}_1 - d\alpha\right)
  = -\sqrt{\frac{d\norm{\tilde{x}}_2^2 - \norm{\tilde{x}}_1^2}{d\lambda_2^2 - \lambda_1^2}}
  < 0\text{.}
\end{displaymath}
With $d\geq 2$ there are two indices $i,j\in I$ with $x_i > x_j$.
The derivative of $\lag$ for $p$ and the complementary slackness condition then yield $p_i - p_j = \frac{1}{\beta}\left(x_i - \alpha - x_j + \alpha\right) = \frac{1}{\beta}\left(x_i - x_j\right) < 0$, which contradicts the order-preservation as guaranteed by Proposition~\ref{prop:orderpres}.
Therefore, the choice of $\alpha$ was not correct in the first place, and thus
\begin{displaymath}
  \alpha = \frac{1}{d}\left(\norm{\tilde{x}}_1 - \lambda_1\sqrt{\frac{d\norm{\tilde{x}}_2^2 - \norm{\tilde{x}}_1^2}{d\lambda_2^2 - \lambda_1^2}}\right)
  \text{ and }
  \beta = \sqrt{\frac{d\norm{\tilde{x}}_2^2 - \norm{\tilde{x}}_1^2}{d\lambda_2^2 - \lambda_1^2}} > 0
\end{displaymath}
must hold.
Let $i\in I$, then $0 < p_i = \frac{1}{\beta}\left(x_i - \alpha\right)$, and because $\beta > 0$ follows $x_i > \alpha$.
For $i\not\in I$ it is $0 = p_i = \frac{1}{\beta}\left(x_i - \alpha + \gamma_i\right)$ where $\gamma_i \geq 0$, so $0 = x_i - \alpha + \gamma_i \geq x_i - \alpha$, or equivalently $x_i \leq \alpha$.
Ultimately, we have that $p_i = \max\big(\frac{1}{\beta}\left(x_i - \alpha\right),\ 0\big)$ for all $i\inint{1}{n}$.

For the claim to hold, it now remains to be shown that $\alpha$ and $\beta$ are unique.
With the uniqueness of the projection $p$, we thus have to show that from
\begin{displaymath}
  p = \max\big(\tfrac{1}{\beta_1}\left(x - \alpha_1\cdot e\right),\ 0\big) = \max\big(\tfrac{1}{\beta_2}\left(x - \alpha_2\cdot e\right),\ 0\big)
\end{displaymath}
for $\alpha_1,\alpha_2,\beta_1,\beta_2\in\R$ follows that $\alpha_1 = \alpha_2$ and $\beta_1 = \beta_2$.
As shown earlier, there are two distinct indices $i,j\in I$ with $x_i\neq x_j$ and $p_i,p_j > 0$.
We hence obtain
\begin{displaymath}
  p_i = \tfrac{1}{\beta_1}\left(x_i - \alpha_1\right) = \tfrac{1}{\beta_2}\left(x_i - \alpha_2\right)\text{ and }
  p_j = \tfrac{1}{\beta_1}\left(x_j - \alpha_1\right) = \tfrac{1}{\beta_2}\left(x_j - \alpha_2\right)\text{,}
\end{displaymath}
and thus $\tfrac{p_i}{p_j} = \tfrac{x_i - \alpha_1}{x_j - \alpha_1} = \tfrac{x_i - \alpha_2}{x_j - \alpha_2}$.
Therefore,
\begin{displaymath}
  0
  = (x_i - \alpha_1)(x_j - \alpha_2) - (x_i - \alpha_2)(x_j - \alpha_1)
  = \alpha_1(x_i - x_j) - \alpha_2(x_i - x_j)
  = (\alpha_1 - \alpha_2)(x_i - x_j)\text{.}
\end{displaymath}
With $x_i\neq x_j$ we have that $\alpha_1 = \alpha_2$, and substitution in either of $p_i$ or $p_j$ shows that $\beta_1 = \beta_2$.
\end{proof}
We note that the crucial point in the computation of $\alpha$ is finding the set $I$ where the projection has positive coordinates.
With the statement of Proposition~\ref{prop:orderpres} the argument of the projection can be sorted before-hand such that $I = \discint{1}{d}$ and therefore only a number linear in $n$ of feasible index sets has to be checked.
This is essentially the method proposed by~\cite{Potluru2013}.
The drawback of this approach is that the time complexity is quasilinear in $n$ because of the sorting, and the space complexity is linear in $n$ because the permutation has to be remembered to be undone afterwards.

\section{Finding the Zero of the Auxiliary Function}
\label{sect:auxfunc}
Lemma~\ref{lem:representation} gives a compact expression that characterizes projections onto $D$.
We first note that the representation only depends on one number:
\begin{remark}
\label{rem:proj1d}
Let $x\in\R_{\geq 0}^n\setminus D$ such that $\proj_D(x)$ is unique.
Then there is exactly one $\alpha\in\R$ such that $\proj_D(x) = \frac{\lambda_2\cdot\max\left(x - \alpha\cdot e,\ 0\right)}{\norm{\max\left(x - \alpha\cdot e,\ 0\right)}_2}$.
\end{remark}
\begin{proof}
The projection becomes $\proj_D(x) = \max\big(\frac{1}{\beta}\left(x - \alpha\cdot e\right),\ 0\big)$ with unique numbers $\alpha\in\R$ and $\beta\in\R_{>0}$ due to Lemma~\ref{lem:representation}.
With $\beta > 0$ we have that $\lambda_2 = \bnorm{\max\big(\frac{1}{\beta}\left(x - \alpha\cdot e\right),\ 0\big)}_2 = \frac{1}{\beta}\norm{\max\left(x - \alpha\cdot e,\ 0\right)}_2$, and the claim follows.
\end{proof}
It can hence be concluded that the sparseness projection can be considered a soft variant of thresholding~\cite{Hyvaerinen1999a}:
\begin{definition}
The function $\shrink_\alpha\colon\R\to\R$, $x\mapsto\max(x - \alpha,\ 0)$, is called \emph{soft-shrinkage function}, where $\alpha\in\R$.
It is continuous on $\R$ and differentiable exactly on $\R\setminus\set{\alpha}$.
\end{definition}
With Remark~\ref{rem:proj1d} we know that we only have to find one scalar to compute projections onto $D$. Analogous to the projection onto a simplex~\cite{Liu2009a}, we can thus define an auxiliary function which vanishes exactly at the number that yields the projection:
\begin{definition}
Let $x\in\R_{\geq 0}^n\setminus D$ such that $\proj_D(x)$ is unique and $\sigma(x) < \sigma^*$.
Let the maximum entry of $x$ be denoted by $x_{\max}:= \max_{i\inint{1}{n}}x_i$.
Then the function
\begin{displaymath}
  \Psi\colon\intervalco{0}{x_{\max}}\to\R\text{,}\qquad \alpha\mapsto\frac{\norm{\max\left(x - \alpha\cdot e,\ 0\right)}_1}{\norm{\max\left(x - \alpha\cdot e,\ 0\right)}_2} - \frac{\lambda_1}{\lambda_2}
\end{displaymath}
is called \emph{auxiliary function} to the projection onto $D$.
\end{definition}
Note that the case of $\sigma(x) \geq \sigma^*$ is trivial, because in this sparseness-decreasing setup we have that all coordinates of the projection must be positive.
Hence $I = \discint{1}{n}$ in the proof of Lemma~\ref{lem:representation}, and the shifting scalar $\alpha$ can be computed from a closed-form expression.

We further fix some notation for convenience:
\begin{definition}
Let $x\in\R_{\geq 0}^n$ be a vector.
Then we write $\X := \set{x_i | i\inint{1}{n}}$ for the set of entries of $x$.
Further, let $x_{\min} := \min\X$ be short for the smallest entry of $x$, and $x_{\max} := \max\X$ and $x_{\scndmax} := \max\X\setminus \set{x_{\max}}$ denote the two largest entries of $x$.

Let $q\colon\R\to\R^n$, $\alpha\mapsto \max\left(x - \alpha\cdot e,\ 0\right)$, denote the curve that evolves from entry-wise application of the soft-shrinkage function.
Let the Manhattan norm and Euclidean norm of points from $q$ be given by $\ell_1\colon\R\to\R$, $\alpha\mapsto\norm{q(\alpha)}_1$, and $\ell_2\colon\R\to\R$, $\alpha\mapsto\norm{q(\alpha)}_2$, respectively.
We thus have that $\Psi = \frac{\ell_1}{\ell_2} - \frac{\lambda_1}{\lambda_2}$, and we find that $q(\alpha)\neq 0$ if and only if $\alpha < x_{\max}$.
\end{definition}
In order to efficiently find the zeros of $\Psi$ we first investigate its analytical properties.
See Figure~\ref{fig:projfunclin-toy} for an example of $\Psi$ that provides orientation for the next result.
\begin{figure}[t]
  \includegraphics[width=\textwidth]{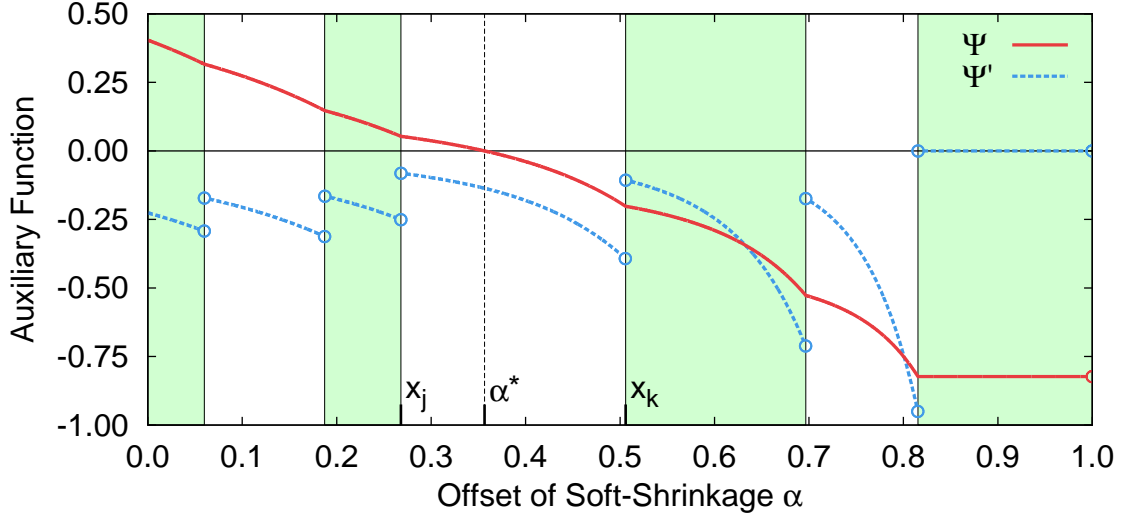}
  \caption{Plot of the auxiliary function $\Psi$ and its derivative for a random vector $x$ (see Lemma~\ref{lem:auxprops} for an analysis). The derivative $\Psi'$ was scaled using a positive number for improved visibility. The steps in $\Psi'$ are exactly the places where $\alpha$ coincides with an entry of $x$. With Remark~\ref{rem:psi_neighbors}, it is sufficient to find an $\alpha$ such that $\Psi(x_j) \geq 0$ and $\Psi(x_k) < 0$ for the neighboring entries $x_j$ and $x_k$ in $x$, because then the exact solution $\alpha^*$ can be computed with a closed-form expression.}
  \label{fig:projfunclin-toy}
\end{figure}

\begin{lemma}
\label{lem:auxprops}
Let $x\in\R_{\geq 0}^n\setminus D$ be given such that the auxiliary function $\Psi$ is well-defined.
Then:
\begin{enumerate}
  \item \label{lem:auxprops_a} $\Psi$ is continuous on $\intervalco{0}{x_{\max}}$.
  \item \label{lem:auxprops_b} $\Psi$ is differentiable on $\intervalco{0}{x_{\max}}\setminus\X$.
  \item \label{lem:auxprops_c} $\Psi$ is strictly decreasing on $\intervalco{0}{x_{\scndmax}}$ and constant on $\intervalco{x_{\scndmax}}{x_{\max}}$.
  \item \label{lem:auxprops_d} There is exactly one $\alpha^*\in\intervaloo{0}{x_{\scndmax}}$ with $\Psi(\alpha^*) = 0$.
  \item \label{lem:auxprops_e} $\proj_D(x) = \frac{\lambda_2\cdot\max\left(x - \alpha^*\cdot e,\ 0\right)}{\norm{\max\left(x - \alpha^*\cdot e,\ 0\right)}_2}$ where $\alpha^*$ is the zero of $\Psi$.
\end{enumerate}
\end{lemma}
\begin{proof}
\ref{lem:auxprops_a}
$q$ is continuous because the soft-shrinkage function is continuous.
Hence so are $\ell_1$ and $\ell_2$, and hence $\Psi$ as compositions of continuous functions.

\ref{lem:auxprops_b}
The soft-shrinkage function causes $\Psi$ to be differentiable exactly on $\intervalco{0}{x_{\max}}\setminus\X$.
Now let $x_j < x_k$ be two successive elements from $\X$, such that there is no element from $\X$ between them.
In the case that $x_k = x_{\min}$ it can be assumed that $x_j = 0$.
Then the index set $I := \set{i\inint{1}{n} | x_i > \alpha}$ of non-vanishing coordinates in $q$ is constant for $\alpha\in\intervaloo{x_j}{x_k}$, and the derivative of $\Psi$ can be computed using a closed-form expression.
For this, let $d := \abs{I}$ denote the number of nonzero entries in $q$.
With $\ell_1(\alpha) = \sum_{i\in I}\left(x_i - \alpha\right) = \sum_{i\in I}x_i - d\alpha$ we obtain $\ell_1'(\alpha) = -d$.
Analogously, it is $\frac{\partial}{\partial\alpha}\ell_2(\alpha)^2 = \frac{\partial}{\partial\alpha}\sum_{i\in I}\left(x_i - \alpha\right)^2 = -2 \sum_{i\in I}\left(x_i - \alpha\right) = -2\ell_1(\alpha)$,
and hence $\ell_2'(\alpha) = \frac{\partial}{\partial\alpha}\sqrt{\ell_2(\alpha)^2} = \frac{1}{2}\ell_2(\alpha)^{-1}\frac{\partial}{\partial\alpha}\ell_2(\alpha)^2 = -\frac{\ell_1(\alpha)}{\ell_2(\alpha)}$.
Therefore, the quotient rule yields
\begin{displaymath}
  \Psi'(\alpha)
  = \frac{-d\ell_2(\alpha) + \ell_1(\alpha)\frac{\ell_1(\alpha)}{\ell_2(\alpha)}}{\ell_2(\alpha)^2}
  = \frac{1}{\ell_2(\alpha)}\left(\frac{\ell_1(\alpha)^2}{\ell_2(\alpha)^2} - d\right)\text{.}
\end{displaymath}
It can further be shown that higher derivatives are of similar form.
We have that $\frac{\partial}{\partial\alpha}\ell_1(\alpha)^2 = 2\ell_1(\alpha)\ell_1'(\alpha) = -2d\ell_1(\alpha)$, and thus
\begin{displaymath}
  \frac{\partial}{\partial\alpha}\frac{\ell_1(\alpha)^2}{\ell_2(\alpha)^2}
  = \frac{-2d\ell_1(\alpha)\ell_2(\alpha)^2 + 2\ell_1(\alpha)^ 3}{\ell_2(\alpha)^4}
  = 2\frac{\ell_1(\alpha)}{\ell_2(\alpha)^2}\left(\frac{\ell_1(\alpha)^2}{\ell_2(\alpha)^2} - d\right)\text{.}
\end{displaymath}
We also obtain $\frac{\partial}{\partial\alpha}\frac{1}{\ell_2(\alpha)} = \frac{-\ell_2'(\alpha)}{\ell_2(\alpha)^2} = \frac{\ell_1(\alpha)}{\ell_2(\alpha)^3}$, and eventually
\begin{displaymath}
  \Psi''(\alpha)
  = \frac{\ell_1(\alpha)}{\ell_2(\alpha)^3}\left(\frac{\ell_1(\alpha)^2}{\ell_2(\alpha)^2} - d\right)
    + \frac{2}{\ell_2(\alpha)}\frac{\ell_1(\alpha)}{\ell_2(\alpha)^2}\left(\frac{\ell_1(\alpha)^2}{\ell_2(\alpha)^2} - d\right)
  = 3\frac{\ell_1(\alpha)}{\ell_2(\alpha)^3}\left(\frac{\ell_1(\alpha)^2}{\ell_2(\alpha)^2} - d\right)\text{,}
\end{displaymath}
or in other words $\frac{\Psi''(\alpha)}{\Psi'(\alpha)} = 3 \frac{\ell_1(\alpha)}{\ell_2(\alpha)^2}$.

\ref{lem:auxprops_c}
First let $\alpha\in\intervaloo{x_{\scndmax}}{x_{\max}}$.
With the notation of~\ref{lem:auxprops_b} we then have that $d = 1$, such that $q$ has exactly one non-vanishing coordinate.
Hence, $\ell_1(\alpha) = \ell_2(\alpha)$ and $\Psi'\equiv 0$ on $\intervaloo{x_{\scndmax}}{x_{\max}}$, thus $\Psi$ is constant on $\intervaloo{x_{\scndmax}}{x_{\max}}$ as a consequence of the mean value theorem from real analysis.
Because $\Psi$ is continuous, it is constant even on $\intervalco{x_{\scndmax}}{x_{\max}}$.

Next let $\alpha\in\intervalco{0}{x_{\scndmax}}\setminus\X$, then $d\geq 2$ and $\ell_1(\alpha) \leq \sqrt{d}\ell_2(\alpha)$ with Remark~\ref{rem:norms_equivalent}.
The inequality is in fact strict, because $q(\alpha)$ has at least two distinct nonzero entries.
This implies that $\Psi' < 0$ on $\intervaloo{x_j}{x_k}$ where $x_j < x_k$ are neighbors of $\alpha$ as in~\ref{lem:auxprops_b}.
The mean value theorem then guarantees that $\Psi$ is strictly decreasing between neighboring elements from $\X$.
This property holds then for the entire interval $\intervalco{0}{x_{\scndmax}}$ due to the continuity of $\Psi$.

\ref{lem:auxprops_d}
We have by requirement that $\sigma(x) < \sigma^*$, and therefore $\frac{\norm{x}_1}{\norm{x}_2} > \frac{\lambda_1}{\lambda_2}$, and so $\Psi(0) > 0$.
For $\alpha\in\intervaloo{x_{\scndmax}}{x_{\max}}$ we obtain $\ell_1(\alpha) = \ell_2(\alpha)$ as in~\ref{lem:auxprops_c}.
It is then $\Psi(\alpha) < 0$ using $\lambda_2 < \lambda_1$.
The existence of $\alpha^*\in\intervalco{0}{x_{\scndmax}}$ with $\Psi(\alpha^*) = 0$ follows from the intermediate value theorem and~\ref{lem:auxprops_c}.
Uniqueness of $\alpha^*$ is guaranteed because $\Psi$ is strictly monotone.

\ref{lem:auxprops_e}
With Remark~\ref{rem:proj1d} there is exactly one $\tilde{\alpha}\in\R$ such that $\proj_D(x) = \frac{\lambda_2\cdot\max\left(x - \tilde{\alpha}\cdot e,\ 0\right)}{\norm{\max\left(x - \tilde{\alpha}\cdot e,\ 0\right)}_2}$.
We see that $\Psi(\tilde{\alpha}) = 0$, and the uniqueness of the zero of $\Psi$ implies that $\alpha^* = \tilde{\alpha}$.
\end{proof}
The unique zero of the auxiliary function can be found numerically using standard root-finding algorithms, such as Bisection, Newton's method or Halley's method~\cite{Traub1964}.
We can improve on this by noting that whenever a number is found in a certain interval, then the exact value of the zero of $\Psi$ can already be computed.
\begin{remark}
\label{rem:interval_existence}
Let $x\in\R_{\geq 0}^n\setminus D$ such that the auxiliary function $\Psi$ is well-defined.
Then there are two unique numbers $x_j < x_k$, where either $x_j = 0$ and $x_k = x_{\min}$ or $x_j,x_k\in\X$ such that there is no other element from $\X$ in between, such that $\Psi(x_j)\geq 0$ and $\Psi(x_k) < 0$ and there is an $\alpha^*\in\intervalco{x_j}{x_k}$ with $\Psi(\alpha^*) = 0$.
\end{remark}
\begin{proof}
Let $\alpha^*\in\intervaloo{0}{x_{\scndmax}}$ with $\Psi(\alpha^*) = 0$ be given with Lemma~\ref{lem:auxprops}.
When $\alpha^* < x_{\min}$ holds, existence follows immediately with Lemma~\ref{lem:auxprops} by setting $x_j := 0$ and $x_k := x_{\min}$.
Otherwise, define $x_j := \max\set{x_i | x_i\in\X\text{ and }x_i\leq\alpha^*}$ and $x_k := \min\set{x_i | x_i\in\X\text{ and }x_i > \alpha^*}$, 
which both exist as the sets where the maximum and the minimum is taken are nonempty.
Clearly these two numbers fulfill the condition from the claim by Lemma~\ref{lem:auxprops}.
The bracketing by $x_j$ and $x_k$ is unique because $\alpha^*$ is in both cases unique with Lemma~\ref{lem:auxprops}.
\end{proof}
We further note that it is easy to check whether the correct interval has already been found, and give a closed-form expression for the zero of $\Psi$ in this case:
\begin{remark}
\label{rem:psi_neighbors}
Let $x\in\R_{\geq 0}^n\setminus D$ such that the auxiliary function $\Psi$ is well-defined and let $\alpha\in\intervalco{0}{x_{\max}}$.
If $\alpha < x_{\min}$ define $x_j := 0$ and $x_k := x_{\min}$, otherwise let $x_j \leq \alpha < x_k$ with $x_j,x_k\in\X$ such that there is no element from $\X$ between $x_j$ and $x_k$.
Let $I := \set{i\inint{1}{n} | x_i > \alpha} = \discint{i_1}{i_d}$ where $d := \abs{I}$ and $\tilde{x}\in\R_{\geq 0}^d$ such that $\tilde{x}\transp = \left(x_{i_1},\dots,x_{i_d}\right)$.
Then the following holds:
\begin{enumerate}
  \item \label{rem:psi_neighbors_a} $\ell_1(\xi) = \norm{\tilde{x}}_1 - d\xi$ and $\ell_2^2(\xi) = \norm{\tilde{x}}_2^2 - 2\xi\norm{\tilde{x}}_1 + d\xi^2$ for $\xi\in\set{x_j,\ \alpha,\ x_k}$.
  \item \label{rem:psi_neighbors_b} When $\lambda_2\ell_1(x_j) \geq \lambda_1\ell_2(x_j)$ and $\lambda_2\ell_1(x_k) < \lambda_1\ell_2(x_k)$ hold, then
\begin{displaymath}
  \alpha^* := \frac{1}{d}\left(\norm{\tilde{x}}_1 - \lambda_1\sqrt{\frac{d\norm{\tilde{x}}_2^2 - \norm{\tilde{x}}_1^2}{d\lambda_2^2 - \lambda_1^2}}\right)
\end{displaymath}
is the unique zero of $\Psi$.
\end{enumerate}
\end{remark}
\begin{proof}
\ref{rem:psi_neighbors_a}
We have that $\ell_1(\alpha) = \sum_{i\in I}(x_i - \alpha) = \sum_{i\in I}x_i - d\alpha = \norm{\tilde{x}}_1 - d\alpha$ and further $\ell_2(\alpha)^2 = \sum_{i\in I}(x_i - \alpha)^2 = \sum_{i\in I}\left(x_i^2 - 2\alpha x_i + \alpha^2\right) = \norm{\tilde{x}}_2^2 - 2\alpha\norm{\tilde{x}}_1 + d\alpha^2$.

Now let $K := \set{i\inint{1}{n} | x_i > x_k}$ and $\tilde{K} := \set{i\inint{1}{n} | x_i = x_k}$.
Then $K = I \setminus \tilde{K}$, and thus $\ell_1(x_k) = \sum_{i\in K}(x_i - x_k) = \sum_{i\in I}(x_i - x_k) - \sum_{i\in\tilde{K}}(x_i - x_k) = \sum_{i\in I}(x_i - x_k) = \norm{\tilde{x}}_1 - d x_k$.
Likewise follows $\ell_2(x_k)^2 = \norm{\tilde{x}}_2^2 - 2 x_k\norm{\tilde{x}}_1 + d x_k^2$.

Finally, let $J := \set{i\inint{1}{n} | x_i > x_j}$ and $\tilde{J} := \set{i\inint{1}{n} | x_i = x_j}$.
Then $I = J\setminus\tilde{J}$, and we obtain $\ell_1(x_j) = \sum_{i\in J}(x_i - x_j) = \sum_{i\in I}(x_i - x_j) + \sum_{i\in\tilde{J}}(x_i - x_j) = \sum_{i\in I}(x_i - x_j) = \norm{\tilde{x}}_1 - d x_j$.
The claim for $\ell_2(x_j)^2$ follows analogously.

\ref{rem:psi_neighbors_b}
The condition from the claim is equivalent to $\Psi(x_j) \geq 0$ and $\Psi(x_k) < 0$.
Hence with Remark~\ref{rem:interval_existence} there is an $\alpha^*\in\intervalco{x_j}{x_k}$ with $\Psi(\alpha^*) = 0$.
Let $p := \proj_D(x)$ be the projection of $x$ onto $D$ and define $J := \set{i\inint{1}{n} | p_i > 0}$.
Lemma~\ref{lem:auxprops}\ref{lem:auxprops_e} implies $J = \set{i\inint{1}{n} | x_i > \alpha^*}$.
Furthermore, it is $J = \set{i\inint{1}{n} | x_i > x_j} = \set{i\inint{1}{n} | x_i > \alpha} = I$.
Thus we already had the correct set of non-vanishing coordinates of the projection in the first place, and the expression for $\alpha^*$ follows from the proof of Lemma~\ref{lem:representation}.
\end{proof}

\section{A Linear Time and Constant Space Projection Algorithm}
By exploiting the analytical properties of $\Psi$, simple methods are sufficient to locate the interval in which its zero resides.
Because the interval has a positive length, simple Bisection is guaranteed to find it in a constant number of steps~\cite{Liu2009a}.
Empirically, we found that solvers that use the derivative of $\Psi$ converge faster, despite of the step discontinuities of $\Psi'$.
We have implemented Newton's method, Halley's method, and Newton's method applied to the slightly transformed auxiliary function $\tilde{\Psi} := \frac{\ell_1^2}{\ell_2^2} - \frac{\lambda_1^2}{\lambda_2^2}$.
These methods were additionally safeguarded with Bisection to guarantee new positions are located within well-defined bounds~\cite{Press2007}.
This does impair the theoretical property that only a constant number of steps be required to find a solution, but in practice a significantly smaller number of steps needs to be made compared to plain Bisection.
This is demonstrated through experimental results in Section~\ref{sect:experiments}.

We are now in a position to formulate the main result of this paper, by proposing an efficient algorithm for computing sparseness-enforcing projections:
\begin{theorem}
Algorithm~\ref{alg:projfunc} computes projections onto $D$, where unique, in a number of operations linear in the problem dimensionality $n$ and with only constant additional space.
\end{theorem}
The proof is omitted as it is essentially a composition of the results from Section~\ref{sect:auxfunc}.

\begin{algorithm}[t]
  \caption{Linear time and constant space evaluation of the auxiliary function $\Psi$.}
  \label{alg:auxiliary}
  \SetAlgoLined
  \KwIn{$x\in\R_{\geq 0}^n$, $\lambda_1, \lambda_2, \alpha\in\R$ with $0 < \lambda_2 < \lambda_1 < \sqrt{n}\lambda_2$ and $0 \leq \alpha < \max_{i\inint{1}{n}}x_i$.}
  \KwOut{$\Psi(\alpha), \Psi'(\alpha), \Psi''(\alpha), \tilde{\Psi}(\alpha), \tilde{\Psi}'(\alpha)\in\R$, $\finished\in\B$, $\ell_1, \ell_2^2\in\R$, $d\in\N$.}
  \BlankLine

  \tcp{Initialize.}
  $\ell_1 := 0$;\enspace $\ell_2^2 := 0$;\enspace $d := 0$;\enspace $x_j := 0$;\enspace $\Delta x_j := -\alpha$;\enspace $x_k := \infty$;\enspace $\Delta x_k := \infty$\;
  \BlankLine

  \tcp{Scan through $x$.}
  \For{$i := 1$ \KwTo $n$}
  {%
    $t := x_i - \alpha$\;
    \eIf{$t > 0$}
    {%
      $\ell_1 := \ell_1 + x_i$;\enspace $\ell_2^2 := \ell_2^2 + x_i^2$;\enspace $d := d + 1$\;
      \lIf{$t < \Delta x_k$}
      {%
        $x_k := x_i$;\enspace $\Delta x_k := t$;
      }
    }
    {%
      \lIf{$t > \Delta x_j$}
      {%
        $x_j := x_i$;\enspace $\Delta x_j := t$;
      }
    }
  }
  \BlankLine

  \tcp{Compute $\Psi(\alpha)$, $\Psi'(\alpha)$ and $\Psi''(\alpha)$.}
  $\ell_1(\alpha) := \ell_1 - d\alpha$;\enspace $\ell_2(\alpha)^2 := \ell_2^2 - 2\alpha\ell_1 + d\alpha^2$\;
  $\Psi(\alpha) := \frac{\ell_1(\alpha)}{\sqrt{\ell_2(\alpha)^2}} - \frac{\lambda_1}{\lambda_2}$;\enspace $\Psi'(\alpha) := \frac{1}{\sqrt{\ell_2(\alpha)^2}}\left(\frac{\ell_1(\alpha)^2}{\ell_2(\alpha)^2} - d\right)$;\enspace $\Psi''(\alpha) := \frac{3\Psi'(\alpha)\ell_1(\alpha)}{\ell_2(\alpha)^2}$\;
  \BlankLine

  \tcp{Compute $\tilde{\Psi}(\alpha)$ and $\tilde{\Psi}'(\alpha)$.}
  $\tilde{\Psi}(\alpha) := \frac{\ell_1(\alpha)^2}{\ell_2(\alpha)^2} - \frac{\lambda_1^2}{\lambda_2^2}$;\enspace $\tilde{\Psi}'(\alpha) := 2\frac{\ell_1(\alpha)}{\sqrt{\ell_2(\alpha)^2}}\Psi'(\alpha)$\;
  \BlankLine

  \tcp{Compute $\Psi(x_j)$ and $\Psi(x_k)$, check for sign change and return.}
  $\finished := \lambda_2(\ell_1 - dx_j) \geq \lambda_1\sqrt{\ell_2^2 - 2x_j\ell_1 + dx_j^2}$\enspace{\bf and}\enspace$\lambda_2\left(\ell_1 - dx_k\right) < \lambda_1\sqrt{\ell_2^2 - 2x_k\ell_1 + dx_k^2}$\;

  \KwRet $\left(\Psi(\alpha),\ \Psi'(\alpha),\ \Psi''(\alpha),\ \tilde{\Psi}(\alpha),\ \tilde{\Psi}'(\alpha),\ \finished,\ \ell_1,\ \ell_2^2,\ d\right)$\;
\end{algorithm}

\begin{algorithm}[p]
  \caption{Linear time and constant space projection onto $D$. The auxiliary function $\Psi$ is evaluated by calls to "$\aux$", which are carried out by Algorithm~\ref{alg:auxiliary}.}
  \label{alg:projfunc}
  \SetAlgoLined
  \SetKw{KwGoTo}{go to}
  \KwIn{$x\in\R_{\geq 0}^n$, $\lambda_1, \lambda_2\in\R$ with $0 < \lambda_2 < \lambda_1 < \sqrt{n}\lambda_2$, $\solver\in\set{\bisection,\ \newton,\ \newtonsqr,\ \halley}$.}
  \KwOut{$\proj_D(x)\in D$ where $D := S_{\geq 0}^{(\lambda_1,\lambda_2)}\subseteq\R_{\geq 0}^n$.}
  \BlankLine

  \tcp{Check whether sparseness should be increased or decreased.}
  $\left(\Psi(\alpha),\ \Psi'(\alpha),\ \Psi''(\alpha),\ \tilde{\Psi}(\alpha),\ \tilde{\Psi}'(\alpha),\ \finished,\ \ell_1,\ \ell_2^2,\ d\right) := \aux(x,\ \lambda_1,\ \lambda_2,\ 0)$\;
  \lIf(\quad\CommentSty{// Decrease sparseness, skip root-finding.}){$\Psi(\alpha) \leq 0$}{\KwGoTo Line~\ref{algl:projfunc-exact-alpha}}
  \BlankLine

  \tcp{Need to increase sparseness, initialize safeguarded root-finding.}
  $\lo := 0$;\enspace $\up := \max\set{x_i | i\inint{1}{n},\ x_i\neq\max_{j\inint{1}{n}}x_j}$;\enspace $\alpha := \lo + \frac{1}{2}(\up - \lo)$\;
  $\left(\Psi(\alpha),\ \Psi'(\alpha),\ \Psi''(\alpha),\ \tilde{\Psi}(\alpha),\ \tilde{\Psi}'(\alpha),\ \finished,\ \ell_1,\ \ell_2^2,\ d\right) := \aux(x,\ \lambda_1,\ \lambda_2,\ \alpha)$\;
  \BlankLine
  \tcp{Perform root-finding until correct interval has been found.}
  \While{{\bf not} $\finished$}
  {%
    \tcp{Update Bisection interval.}
    \leIf{$\Psi(a) > 0$}{$\lo := \alpha$}{$\up := \alpha$}
    \tcp{One iteration of root-finding.}
    \lIf{$\solver = \bisection$}{$\alpha := \lo + \frac{1}{2}(\up - \lo)$}
    \Else(\CommentSty{// Use solvers based on derivatives.})
    {%
      \lIf{$\solver = \newton$}{$\alpha := \alpha - \frac{\Psi(\alpha)}{\Psi'(\alpha)}$}
      \lElseIf{$\solver = \newtonsqr$}{$\alpha := \alpha - \frac{\tilde{\Psi}(\alpha)}{\tilde{\Psi}'(\alpha)}$}
      \ElseIf{$\solver = \halley$}
      {%
        $h := 1 - \frac{\Psi(\alpha)\Psi''(\alpha)}{2\Psi'(\alpha)^2}$;\enspace $h := \max(0.5,\ \min(1.5,\ h))$;\enspace $\alpha := \alpha - \frac{\Psi(\alpha)}{h\Psi'(\alpha)}$\;
      }
      \tcp{If $\alpha$ fell out of bounds, perform normal Bisection.}
      \lIf{$\alpha < \lo$ {\bf or} $\alpha > \up$}{$\alpha := \lo + \frac{1}{2}(\up - \lo)$}
    }
    \tcp{Re-evaluate auxiliary function at new position.}
    $\left(\Psi(\alpha),\ \Psi'(\alpha),\ \Psi''(\alpha),\ \tilde{\Psi}(\alpha),\ \tilde{\Psi}'(\alpha),\ \finished,\ \ell_1,\ \ell_2^2,\ d\right) := \aux(x,\ \lambda_1,\ \lambda_2,\ \alpha)$\;
  }
  \BlankLine

  \tcp{Correct interval has been found, compute exact value for $\alpha$.}
  $\alpha := \dfrac{1}{d}\left(\ell_1 - \lambda_1\sqrt{\dfrac{d\ell_2^2 - \ell_1^2}{d\lambda_2^2 - \lambda_1^2}}\right)$\nllabel{algl:projfunc-exact-alpha};
  \BlankLine

  \tcp{Compute result of the projection in-place.}
  $\rho := 0$\;
  \For{$i := 1$ \KwTo $n$}
  {%
    $t := x_i - \alpha$\;
    \leIf{$t > 0$}{$x_i := t$;\enspace $\rho := \rho + t^2$}{$x_i := 0$}
  }
  \lFor{$i := 1$ \KwTo $n$} { $x_i := \frac{\lambda_2}{\sqrt{\rho}} x_i$}

  \KwRet $x$\;
\end{algorithm}

\section{Gradient of the Projection}
We conclude our analysis of the sparseness-enforcing projection by considering its gradient:
\begin{lemma}
\label{lem:projfuncgrad}
The projection onto $D$ can be cast as function $\R_{\geq 0}^n\to D$ in all points with unique projections, that is almost everywhere.
Further, this function is differentiable almost everywhere.

More precisely, let $x\in\R_{\geq 0}^n\setminus D$ such that $p := \proj_D(x)$ is unique.
With Remark~\ref{rem:proj1d}, let $\alpha\in\R$ such that $p = \frac{\lambda_2\cdot\max\left(x - \alpha\cdot e,\ 0\right)}{\norm{\max\left(x - \alpha\cdot e,\ 0\right)}_2}$.
If $x_i\neq\alpha$ for all $i\inint{1}{n}$, then $\proj_D$ is differentiable in $x$.
It is further possible to give a closed-form expression for the gradient as follows.
Let the index set of nonzero entries in the projection be denoted by $I := \set{i\inint{1}{n} | p_i > 0} = \set{i_1,\dots,i_d}$ where $d := \abs{I}$.
Let $e_k\in\R^n$ denote the $k$th canonical basis vector for $k\inint{1}{n}$ and let $V := \left(e_{i_1},\ \dots,\ e_{i_d}\right)\transp\in\set{0,1}^{d\times n}$ be the slicing matrix with respect to $I$, that is with $\tilde{x} := Vx\in\R^d$ we have for example $\tilde{x}\transp = \left(x_{i_1},\ \dots,\ x_{i_d}\right)$.
Write $a := d\norm{\tilde{x}}_2^2 - \norm{\tilde{x}}_1^2\in\R_{\geq 0}$ and $b := d\lambda_2^2 - \lambda_1^2\in\R_{\geq 0}$ for short.
With Lemma~\ref{lem:representation} we find that $\alpha = \frac{1}{d}\left(\norm{\tilde{x}}_1 - \lambda_1\sqrt{\frac{a}{b}}\right)$.
Denote by $\tilde{e} := Ve\in\set{1}^d$ the vector where all $d$ entries are equal to unity, and let $\tilde{q} := \max\left(\tilde{x} - \alpha\cdot \tilde{e},\ 0\right) = \tilde{x} - \alpha\cdot \tilde{e}\in\R_{\geq 0}^d$ which implies that $p = \frac{\lambda_2}{\norm{\tilde{q}}_2} V\transp\tilde{q}$ holds.

Let $\tilde{p} := \frac{\lambda_2}{\norm{\tilde{q}}_2}\tilde{q}$ such that $p = V\transp\tilde{p}$, then $\frac{\partial}{\partial x}\proj_D(x) = V\transp G V\in\R^{n\times n}$ where
\begin{displaymath}
  G := \sqrt{\tfrac{b}{a}}E_d - \tfrac{1}{\sqrt{ab}}\left(\lambda_2^2\tilde{e}\tilde{e}\transp + d\tilde{p}\tilde{p}\transp - \lambda_1\left(\tilde{e}\tilde{p}\transp + \tilde{p}\tilde{e}\transp\right)\right)\text{.}
\end{displaymath}
\end{lemma}
\begin{proof}
The projection is unique almost everywhere as already shown by~\cite{Theis2005}.
When $x_i\neq\alpha$ for all $i\inint{1}{n}$, $\proj_D$ is differentiable as composition of differentiable functions as then $I$ is invariant to local changes in $x$.
Write $\tilde{p} := \frac{\lambda_2}{\norm{\tilde{q}}_2}\tilde{q}$ such that $p = V\transp\tilde{p}$, then the chain rule yields
\begin{displaymath}
  \frac{\partial p}{\partial x}
  = \frac{\partial V\transp\tilde{p}}{\partial\tilde{p}}\cdot
    \frac{\partial}{\partial\tilde{q}}\left(\frac{\lambda_2}{\norm{\tilde{q}}_2}\tilde{q}\right)\cdot
    \frac{\partial\left(\tilde{x} - \alpha\cdot \tilde{e}\right)}{\partial\tilde{x}}\cdot
    \frac{\partial Vx}{\partial x}
  = V\transp\cdot
    \underbrace{\lambda_2\frac{\partial}{\partial\tilde{q}}\left(\frac{\tilde{q}}{\norm{\tilde{q}}_2}\right)\cdot
    \left(E_d - \tilde{e}\frac{\partial\alpha}{\partial\tilde{x}}\right)}_{=: G\in\R^{d\times d}}\cdot
    V\text{.}
\end{displaymath}
We thus only have to show that the matrix $G$ defined here matches the matrix from the claim.
It is easy to see that the mapping from a vector to its normalized version has a simple gradient, that is we have that
$\frac{\partial}{\partial\tilde{q}}\frac{\tilde{q}}{\norm{\tilde{q}}_2} = \frac{1}{\norm{\tilde{q}}_2}\left(E_d - \frac{\tilde{q}\tilde{q}\transp}{\norm{\tilde{q}}_2^2}\right)$.
Because $\tilde{q}$ and $\tilde{x}$ have only non-negative entries, the canonical dot product with $\tilde{e}$ yields essentially their $L_1$ norms.
We hence obtain $\norm{\tilde{q}}_1 = \tilde{e}\transp\tilde{x} - \alpha\tilde{e}\transp\tilde{e} = \lambda_1\sqrt{\frac{a}{b}}$.
Likewise, the $L_2$ norm of $\tilde{q}$ equals
\begin{align*}
  \norm{\tilde{q}}_2^2
  &= \norm{\tilde{x}}_2^2 - 2\alpha\norm{\tilde{x}}_1 + d\alpha^2
  = \norm{\tilde{x}}_2^2 - \alpha\left(\norm{\tilde{x}}_1 + \lambda_1\sqrt{\tfrac{a}{b}}\right)
  = \norm{\tilde{x}}_2^2 - \tfrac{1}{d}\left(\norm{\tilde{x}}_1^2 - \lambda_1^2\tfrac{a}{b}\right)\\
  &= \tfrac{1}{d}\big(d\norm{\tilde{x}}_2^2 - \norm{\tilde{x}}_1^2\big)+ \lambda_1^2\tfrac{a}{b}
  = \tfrac{a}{d}\left(1 + \tfrac{\lambda_1^2}{b}\right)
  = \lambda_2^2\tfrac{a}{b}\text{.}
\end{align*}
To compute the gradient of $\alpha$, we first note that $b$ does not depend on $\tilde{x}$ but $a$ does.
It is $\frac{\partial}{\partial\tilde{x}} a = 2d\tilde{x}\transp - 2\norm{\tilde{x}}_1\tilde{e}\transp\in\R^{1\times d}$, and hence $\frac{\partial}{\partial\tilde{x}} \sqrt{a} = \frac{1}{\sqrt{a}}\left(d\tilde{x}\transp - \norm{\tilde{x}}_1\tilde{e}\transp\right)\in\R^{1\times d}$.
With $\tilde{x} = \tilde{q} + \alpha\cdot\tilde{e}$ follows $d\tilde{x} - \norm{\tilde{x}}_1\tilde{e} = d\tilde{q} - \lambda_1\sqrt{\frac{a}{b}}\tilde{e}$, and hence
\begin{displaymath}
  \frac{\partial}{\partial\tilde{x}}\alpha
  = \tfrac{1}{d}\tilde{e}\transp - \tfrac{\lambda_1}{d\sqrt{ab}}\left(d\tilde{x}\transp - \norm{\tilde{x}}_1\tilde{e}\transp\right)
  = \left(\tfrac{1}{d} + \tfrac{\lambda_1^2}{db}\right)\tilde{e}\transp - \tfrac{\lambda_1}{\sqrt{ab}}\tilde{q}\transp
  = \tfrac{\lambda_2^2}{b}\tilde{e}\transp - \tfrac{\lambda_1}{\sqrt{ab}}\tilde{q}\transp\in\R^{1\times d}\text{.}
\end{displaymath}
Therefore, by substitution into $G$ and multiplying out we yield
\begin{align*}
  G &= \sqrt{\tfrac{b}{a}}\left(E_d - \tfrac{b}{\lambda_2^2 a}\tilde{q}\tilde{q}\transp\right)
      \left(E_d - \tfrac{\lambda_2^2}{b}\tilde{e}\tilde{e}\transp + \tfrac{\lambda_1}{\sqrt{ab}}\tilde{e}\tilde{q}\transp\right)\\
  &= \sqrt{\tfrac{b}{a}}\left(E_d - \tfrac{\lambda_2^2}{b}\tilde{e}\tilde{e}\transp + \tfrac{\lambda_1}{\sqrt{ab}}\tilde{e}\tilde{q}\transp
        - \tfrac{b}{\lambda_2^2 a}\tilde{q}\tilde{q}\transp + \tfrac{\lambda_1}{\sqrt{ab}}\tilde{q}\tilde{e}\transp - \tfrac{\lambda_1^2}{\lambda_2^2 a}\tilde{q}\tilde{q}\transp \right)\text{,}
\end{align*}
where we have used
$\tilde{q}\tilde{q}\transp\tilde{e}\tilde{e}\transp = \tilde{q}\left(\tilde{q}\transp\tilde{e}\right)\tilde{e}\transp = \lambda_1\sqrt{\frac{a}{b}}\tilde{q}\tilde{e}\transp$
and $\tilde{q}\tilde{q}\transp\tilde{e}\tilde{q}\transp = \lambda_1\sqrt{\frac{a}{b}}\tilde{q}\tilde{q}\transp$.
The claim then follows with $\frac{b}{\lambda_2^2 a} + \frac{\lambda_1^2}{\lambda_2^2 a} = \frac{d}{a}$ and $\tilde{q} = \sqrt{\frac{a}{b}}\tilde{p}$.
\end{proof}
The gradient given in Lemma~\ref{lem:projfuncgrad} has a particular simple form, as it is essentially a scaled identity matrix with additive combination of scaled dyadic products of simple vectors.
In the situation where not the entire gradient but merely its product with an arbitrary vector is required, simple vector operations are already enough to compute the product:
\begin{theorem}
Algorithm~\ref{alg:projgrad} computes the product of the gradient of the sparseness projection with an arbitrary vector in time and space linear in the problem dimensionality $n$.
\end{theorem}
This claim can directly be validated using the expression from the gradient given in Lemma~\ref{lem:projfuncgrad}.

\begin{algorithm}[t]
  \caption{Product of the gradient of the projection onto $D$ with an arbitrary vector.}
  \label{alg:projgrad}
  \SetAlgoLined
  \SetKw{KwGoTo}{go to}
  \KwIn{$y\in\R^n$ and the following results of Algorithm~\ref{alg:projfunc}: $I = \discint{i_1}{i_d}\subseteq\discint{1}{n}$, $d := \abs{I}$, $\tilde{p}\in\R_{\geq 0}^d$, $\lambda_1,\lambda_2\in\R_{>0}$, $a := d\norm{\tilde{x}}_2^2 - \norm{\tilde{x}}_1^2\in\R_{\geq 0}$ and $b := d\lambda_2^2 - \lambda_1^2\in\R_{\geq 0}$.}
  \KwOut{$z := \left(\frac{\partial}{\partial x}\proj_D(x)\right)\cdot y\in \R^n$.}
  \BlankLine

  \tcp{Scan and slice input vector.}
  $\tilde{y}\in\set{0}^d$;\enspace $\sy := 0$;\enspace $\scpy := 0$\;
  \lFor{$i := 1$ \KwTo $d$}{$\sy := \sy + z_{i_j}$;\enspace $\scpy := \scpy + \tilde{p}_j\cdot z_{i_j}$;\enspace $\tilde{z}_j := z_{i_j}$}
  \BlankLine

  \tcp{Compute product with gradient in sliced space.}
  $\tilde{z} := \sqrt{\frac{b}{a}}\tilde{z}$;\enspace%
  $\tilde{z} := \tilde{z} + \frac{1}{\sqrt{ab}}\left(\lambda_1\cdot\sy - d\cdot\scpy\right)\tilde{p}$;\enspace%
  $\tilde{z} := \tilde{z} + \frac{1}{\sqrt{ab}}\left(\lambda_1\cdot\scpy - \lambda_2^2\cdot\sy\right)\tilde{e}$\;
  \BlankLine

  \tcp{Unslice to yield final result.}
  $y\in\set{0}^n$;\enspace \lFor{$i := 1$ \KwTo $d$}{$y_{i_j} := \tilde{z}_j$}
  \KwRet $y$\;
\end{algorithm}

\section{Experiments}
\label{sect:experiments}
To assess the performance of the algorithm we proposed to compute sparseness-enforcing projections, several experiments have been carried out.
As the projection onto $D$ is unique almost everywhere, different approaches must compute the same result except for a null set.
We have compared the results of the algorithm proposed by~\cite{Hoyer2004} with the results of our algorithm for problem dimensionalities $n\inint{2^2}{2^{26}}$ and for target degrees of sparseness $\sigma^*\in\set{0.025,\ 0.050,\ \dots,\ 0.950,\ 0.975}$.
For every combination of $n$ and $\sigma^*$ we have sampled one thousand random vectors, carried out both algorithms, and found that both algorithms produce numerically equal results given the very same input vector.
Moreover, we have numerically verified the gradient of the projection for the same range using the central difference quotient.

Finally, experiments have been conducted to evaluate the choice of the solver for Algorithm~\ref{alg:projfunc}.
We have set the problem dimensionality to $n := 1024$, and then sampled one thousand random vectors for target sparseness degrees of $\sigma^*\in\set{0.200,\ 0.225,\ \dots,\ 0.950,\ 0.975}$.
We have used the very same random vectors as input for all solvers, and counted the number of times the auxiliary function had to be computed until the solution was found.
The results of this experiment are depicted in Figure~\ref{fig:solverplotn}.
While Bisection needs about the same number of evaluations over all sparseness degrees, the solvers based on the derivative of $\Psi$ depend on $\sigma^*$ in their number of evaluations.
This is because their starting value is set to the midpoint of the initial bracket in Algorithm~\ref{alg:projfunc}, and thus their distance to the root of $\Psi$ naturally depends on $\sigma^*$.
The solver that performs best is $\newtonsqr$, that is Newton's method applied to $\tilde{\Psi}$.
It is quite surprising that the methods based on derivatives perform so well, as $\Psi'$ possesses several step discontinuities as illustrated in Figure~\ref{fig:projfunclin-toy}.

In the next experiment, the target sparseness degree was set to $\sigma^* := 0.90$ and the problem dimensionality $n$ was varied in $\discint{2^2}{2^{26}}$.
The results are shown in Figure~\ref{fig:solverplotsigma}.
The number of evaluations Bisection needs in the experiment grows about linearly in $\log(n)$.
Because the expected minimum difference of two distinct entries from a random vector gets smaller when the dimensionality of the random vector is increased, the expected number of function evaluations Bisection requires increases with problem dimensionality.
In either case, the length of the interval that has to be found is always bounded from below by the machine precision such that the number of function evaluations with Bisection is bounded from above.
The methods based on derivatives exhibit sublinear growth, where the solver $\newtonsqr$ is again the best performing one.
Note that the number of iterations it requires decreases when dimensionality is enhanced.
This is because Hoyer's sparseness measure $\sigma$ is not invariant to problem dimensionality, and hence a sparseness of $\sigma^* = 0.90$ has a different notion for $n = 2^{26}$ than for $n = 2^8$.

\begin{figure}[p]
  \includegraphics[width=\textwidth]{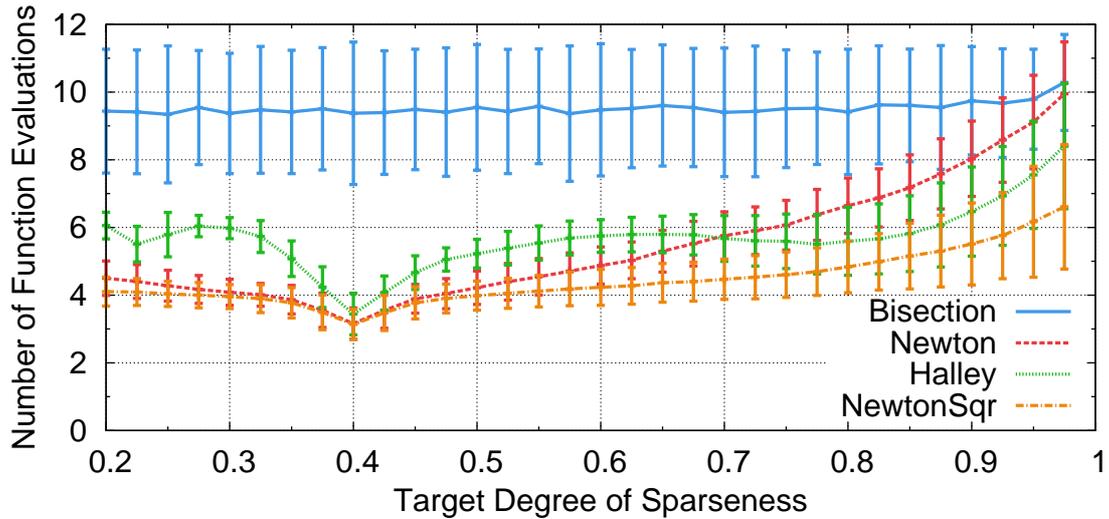}
  \caption{Auxiliary function evaluations needed to find the final interval with four different solvers. The problem dimensionality was set to $n := 1024$ and the target degree of sparseness $\sigma^*$ was varied. While the performance of Bisection is constant over different values of $\sigma^*$, the solvers that use the derivative of the auxiliary function depend on the target sparseness and consistently outperform Bisection. Newton's method applied to $\tilde{\Psi}$ is the best-performing solver over all choices of $\sigma^*$.}
  \label{fig:solverplotn}
\end{figure}

\begin{figure}[p]
  \includegraphics[width=\textwidth]{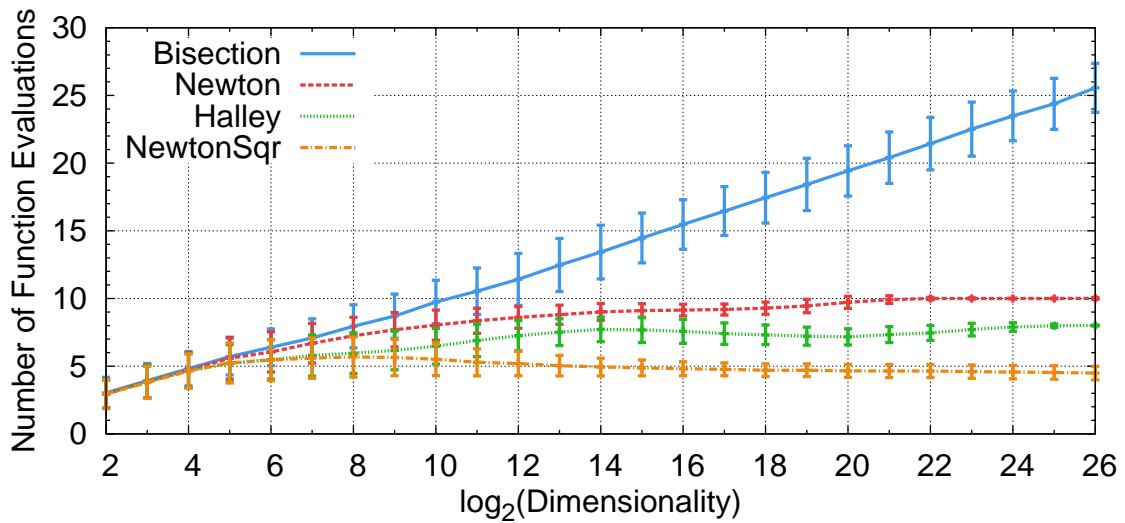}
  \caption{Same plot as in Figure~\ref{fig:solverplotn}, except for the target degree of sparseness was set to $\sigma^* := 0.90$ and the problem dimensionality was varied. The number of required function evaluations grows linearly with the logarithm of the problem dimensionality for Bisection, while the other solvers require a number sublinear in $\log(n)$. When $n = 2^{26}\approx 67\cdot 10^{6}$ then Newton's method only requires $10$~iterations in the mean, and Newton's method applied to $\tilde{\Psi}$ requires only $4$~iterations.}
  \label{fig:solverplotsigma}
\end{figure}

\section{Conclusion}
In this paper, we have proposed an efficient algorithm for computing sparseness-enforcing projections with respect to Hoyer's sparseness measure $\sigma$.
Although the target set of the projection is here non-convex, methods from projections onto simplexes could be adapted in a straightforward way.
We have rigorously proved the correctness of our proposed algorithm, and additionally we have yielded a simple procedure to compute its gradient.
We have shown that our algorithm needs only little resources, and that it scales well with problem dimensionality, even for very high target sparseness degrees.

\subsection*{Acknowledgments}
The authors would like to thank Michael Gabb for helpful discussions.
This work was supported by Daimler~AG, Germany.

\bibliographystyle{IEEEtran}
\bibliography{the}

\begin{thebibliography}{10}
\providecommand{\url}[1]{#1}
\csname url@samestyle\endcsname
\providecommand{\newblock}{\relax}
\providecommand{\bibinfo}[2]{#2}
\providecommand{\BIBentrySTDinterwordspacing}{\spaceskip=0pt\relax}
\providecommand{\BIBentryALTinterwordstretchfactor}{4}
\providecommand{\BIBentryALTinterwordspacing}{\spaceskip=\fontdimen2\font plus
\BIBentryALTinterwordstretchfactor\fontdimen3\font minus
  \fontdimen4\font\relax}
\providecommand{\BIBforeignlanguage}[2]{{%
\expandafter\ifx\csname l@#1\endcsname\relax
\typeout{** WARNING: IEEEtran.bst: No hyphenation pattern has been}%
\typeout{** loaded for the language `#1'. Using the pattern for}%
\typeout{** the default language instead.}%
\else
\language=\csname l@#1\endcsname
\fi
#2}}
\providecommand{\BIBdecl}{\relax}
\BIBdecl

\bibitem{Hurley2009}
N.~Hurley and S.~Rickard, ``{C}omparing measures of sparsity,'' \emph{IEEE
  Transactions on Information Theory}, vol.~55, no.~10, pp. 4723--4741, 2009.

\bibitem{Hoyer2004}
P.~O. Hoyer, ``{N}on-negative matrix factorization with sparseness
  constraints,'' \emph{Journal of Machine Learning Research}, vol.~5, pp.
  1457--1469, 2004.

\bibitem{Bertsekas1999}
D.~P. Bertsekas, \emph{{N}onlinear {P}rogramming}, 2nd~ed.\hskip 1em plus 0.5em
  minus 0.4em\relax Athena Scientific, 1999.

\bibitem{Deutsch2001}
F.~Deutsch, \emph{{B}est {A}pproximation in {I}nner {P}roduct {S}paces}.\hskip
  1em plus 0.5em minus 0.4em\relax Springer, 2001.

\bibitem{Theis2005}
F.~J. Theis, K.~Stadlthanner, and T.~Tanaka, ``{F}irst results on uniqueness of
  sparse non-negative matrix factorization,'' in \emph{Proceedings of the
  European Signal Processing Conference}, 2005, pp. 1672--1675.

\bibitem{Potluru2013}
V.~K. Potluru, S.~M. Plis, J.~L. Roux, B.~A. Pearlmutter, V.~D. Calhoun, and
  T.~P. Hayes, ``{B}lock coordinate descent for sparse {NMF},'' Tech. Rep.
  arXiv:1301.3527v1, 2013.

\bibitem{Liu2009a}
J.~Liu and J.~Ye, ``{E}fficient euclidean projections in linear time,'' in
  \emph{Proceedings of the International Conference on Machine Learning}, 2009,
  pp. 657--664.

\bibitem{Laub2004}
A.~J. Laub, \emph{{M}atrix {A}nalysis for {S}cientists and {E}ngineers}.\hskip
  1em plus 0.5em minus 0.4em\relax Society for Industrial and Applied
  Mathematics, 2004.

\bibitem{Hyvaerinen1999a}
A.~Hyv\"{a}rinen, P.~Hoyer, and E.~Oja, ``{S}parse code shrinkage: {D}enoising
  by nonlinear maximum likelihood estimation,'' in \emph{Advances in Neural
  Information Processing Systems 11}, 1999, pp. 473--478.

\bibitem{Traub1964}
J.~F. Traub, \emph{{I}terative {M}ethods for the {S}olution of
  {E}quations}.\hskip 1em plus 0.5em minus 0.4em\relax Prentice-Hall, 1964.

\bibitem{Press2007}
W.~H. Press, S.~A. Teukolsky, W.~T. Vetterling, and B.~P. Flannery,
  \emph{{N}umerical {R}ecipes: {T}he {A}rt of {S}cientific {C}omputing},
  3rd~ed.\hskip 1em plus 0.5em minus 0.4em\relax Cambridge University Press,
  2007.

\end{thebibliography}

\end{document}